%% file: main.tex
\documentclass[a4paper,USenglish,cleveref, autoref, thm-restate]{lipics-v2021}

\newcommand{\mc}[1]{}

\hideLIPIcs  

\nolinenumbers 

\usepackage[T1]{fontenc}
\usepackage{graphicx} 
\usepackage{color}
\usepackage{hyperref}
\usepackage{amsmath}    
\usepackage{amsfonts}
\usepackage[noend]{algpseudocode}
\usepackage{algorithm}
\usepackage{dirtytalk}
\usepackage{tikz}
\usetikzlibrary{automata, arrows.meta, shapes.geometric, shapes.symbols, fit, backgrounds, positioning, patterns, decorations.pathmorphing, calc}
\usepackage{fontawesome5}
\usepackage{pifont}
\usepackage{xspace}
\usepackage[skins,breakable]{tcolorbox}

\input{macros}

\date{}

\begin{document}
\title{Secret Quorums: Protecting Byzantine Protocols Against Adaptive Adversaries}
\titlerunning{Secret Quorums: Protecting Byzantine Protocols Against Adaptive Adversaries}


\author{Maxence {Perion}}{Université Paris-Saclay, CEA, LIST \and Palaiseau, France}{maxence.perion@cea.fr}{https://orcid.org/0009-0005-4048-6708}{}

\author{Sara {Tucci-Piergiovanni}}{Université Paris-Saclay, CEA, LIST \and Palaiseau, France}{}{https://orcid.org/0000-0001-9738-9021}{}

\author{Rida {Bazzi}}{Arizona State University \and USA}{}{https://orcid.org/0000-0003-1872-1634}{}

\authorrunning{M. Perion, S. Tucci-Piergiovanni and R. Bazzi}

\ccsdesc[300]{Theory of computation~Distributed algorithms}
\ccsdesc[300]{Computer systems organization~Dependable and fault-tolerant systems and networks}
\ccsdesc[300]{Security and privacy~Distributed systems security}
\bibliographystyle{plainurl}

\maketitle

\keywords{Distributed Systems, Blockchain, Quorums, Fault Tolerance, Distributed Asset Transfer}

\begin{abstract}
Modern committee-based payment protocols improve scalability by delegating critical operations to small subsets of participants, such as validator committees in blockchains or shard committees in distributed systems with parallel execution.
This design, however, makes these protocols particularly vulnerable to adaptive adversaries: once a small set of participants is identified, it can be selectively targeted for corruption, bribery, or denial-of-service attacks.

In this paper, we propose \emph{Secret Quorums}, a novel abstraction that enables any committee-based protocol, including payment systems, to rely on small quorums while remaining resilient to adaptive adversaries.
Validators composing a Secret Quorum remain anonymous throughout the validation: as in cryptographic sortition approaches, their selection is secret, but unlike classical approaches, the resulting quorum proof does not reveal which validators were selected.
We show how to implement Secret Quorums using ring verifiable random functions, without adding communication steps compared to standard quorum-based protocols.

We also demonstrate the relevance of Secret Quorums through \protocolname, a new protocol that applies Secret Quorums to the fractional spending payment problem in order to reduce latency  and improve settlement message complexity with respect to the original protocol.

\end{abstract}

\section{Introduction}\label{sec:intro}
Byzantine fault-tolerant protocols typically rely on
\emph{quorum systems}~\cite{byzQuorumSystems,probaQuorums} to enforce safety
and liveness properties.
A quorum is a subset of nodes that a client can query to attest that some data
is valid with respect to the application semantics. For instance, in cryptocurrencies~\cite{nonConsensus},
a quorum consists of $2f+1$ votes from distinct validators attesting that a
transaction does not spend more than the available balance. The number $f$
denotes the maximum number of Byzantine nodes, where a Byzantine node can deviate arbitrarily from its expected behavior, for
instance by validating a double-spending transaction (threatening safety),
or by not responding (threatening liveness).
A quorum ensures that if a Byzantine client tries to overspend through two
different transactions, then the quorums for these transactions intersect in at
least one correct validator, which cannot attest that both transactions are valid.

In a system of $n=3f+1$ nodes, quorums of size $2f+1$ allow protocols to
tolerate up to $f$ Byzantine nodes,  i.e., a system can only tolerate at most one third of its nodes being Byzantine. Thus, tolerating additional failures
requires adding $3$ times more nodes and operating a larger system.
At the same time, involving a large quorum of validators increases the tail latency
required to collect votes. Consequently, Byzantine fault-tolerant protocols are
known to be slow when the number of failures they tolerate is high; in this
sense, they do not scale well.

A widely used technique to improve scalability is to reduce the number of nodes
involved in each validation, compared to standard quorums. This is the approach used in \emph{committee}-based
blockchains (e.g., \cite{EthereumUlysse}, \cite{algorandScalingByzAgree}), where only a subset of size $s$ (a \say{small quorum}) is selected from the whole
population of $n$ nodes to validate transactions. Similar ideas are used to speed
up transaction validation through sharding~\cite{sokSharding,Yggdrasil} or in
two-layer payment approaches (e.g., ~\cite{asynchronousstatechannels, syncPCN}). In all these cases, the fault tolerance of the whole system boils down to
the fault tolerance of the smallest subset of nodes involved in validation, which can be overly conservative when $s \ll n$.

On the other hand, when a large system with many nodes is involved, it becomes
natural to reason about failures in probabilistic terms~\cite{probaQuorums}. Unlike deterministic
fault tolerance, which considers the worst case over all sets of up to $f$
Byzantine nodes, probabilistic approaches consider the probability that a
randomly selected subset of size $s$ contains more corrupted nodes than the
validation protocol can tolerate, while the adversary may still control up to
$f$ nodes in the whole system. This probability can be made sufficiently small
by choosing an appropriate subset size $s$. 
Notably, if the adversarial fraction
$f/n$ decreases, then this probability also
decreases.


These probabilistic guarantees, however, depend on the adversary not being
able to predict or target the selected validators before they participate in the
protocol. This becomes relevant when the adversary is \emph{adaptive}, i.e.,
able to selectively corrupt nodes during the computation, unlike a static
adversary, which fixes the set of corrupted nodes in advance. To protect against
an adaptive adversary in the context of probabilistic, committee-based protocols,
\emph{cryptographic sortition} (also called secret election) has been proposed, formalized~\cite{algorandScalingByzAgree,ssle} and
used in several protocols~\cite{algorand,OuroborosVRF}.
Cryptographic sortition allows nodes to privately determine whether they have
been selected at random. A selected node can then prove its legitimacy using
a cryptographic proof of selection, which is attached to its messages.
Despite the secrecy of the selection, the proof of selection leaks the identity of a selected node, as it is signed with its private key. Thus, selected nodes reveal themselves at the moment they disseminate a vote over the network.

We observe that this can be problematic when the adversary is \textit{strongly adaptive} and \textit{rushing}. By strongly adaptive we mean that it can instantaneously corrupt any node of its choice (up to $f$ nodes). By rushing, we mean that the adversary
may deliver messages to corrupted nodes before delivering them to correct nodes,
observe their content, and then decide its next actions, with the possibility of clawing back the messages intended for correct nodes before the corruption. In particular, if a
selected validator reveals its identity through a vote message, the adversary can learn the validator's identity before
the validation has been sufficiently disseminated among correct nodes. The
adversary can then immediately corrupt the selected validator and erase or modify
its local validation state. From that point on, the validator is no longer
correct, and the reliability guarantees between correct nodes no longer protect
its subsequent messages. As a consequence, the adversary can prevent that
validator from further supporting the validation.
Since the quorum is small, the adversary may repeat this attack on enough selected validators to prevent the validation from being completed, thereby hindering its persistence in the whole system.
In other words, once selected validators reveal their
identities, cryptographic sortition alone is not sufficient to protect them
against a strongly adaptive and rushing adversary.

In this paper, we introduce \emph{Secret Quorums}, a novel abstraction that 
extends the concept of cryptographic sortition~\cite{ssle,algorand} by hiding
the identity of selected validators throughout the quorum validation process.
More specifically, a secret quorum is a quorum of expected size $k$ out of $n$
validators whose membership is \textit{unpredictable}, i.e., the adversary can do no
better than guessing; \textit{random}, i.e., all validators are equally likely to be part
of the quorum; and \textit{anonymous}, i.e., proofs of membership reveal only the
validator's legitimacy, but not its identity. This anonymity property is the key difference with respect to standard cryptographic sortition. 

We then propose an asynchronous protocol that implements the Secret Quorums abstraction using
a cryptographic primitive called \emph{rVRF}~\cite{rVRF}, which combines a
verifiable random function with ring signatures.
We leverage this primitive to allow nodes to privately determine whether they
belong to a quorum, and to produce secrecy-preserving membership proofs using
ring signatures. This construction shows that the Secret Quorums abstraction can
be instantiated using existing cryptographic primitives. We view it as a first
feasible construction, rather than as an optimal one.

Finally, to demonstrate the applicability and benefits of Secret Quorums, we further design \emph{\protocolname}, a
protocol that solves the \emph{fractional spending} problem~\cite{breakingFp1}
while tolerating a strongly adaptive and rushing adversary in an asynchronous network. In this problem, an account with
balance $b$ can issue up to $p$ payments in parallel, without requiring quorum
intersection, provided that each payment spends at most $b/p$. This makes it
possible to validate payments using smaller quorums of size $q<f$. Compared with the original protocol of Bazzi and Tucci~\cite{breakingFp1},
\protocolname is considerably simpler thanks to the use of Secret Quorums.
Moreover, when instantiated with our rVRF-based construction, it achieves a
significant reduction in communication complexity.

\paratitle{Contributions}
Our contributions are threefold:
\begin{itemize}
\item \textbf{Secret Quorums.} We introduce the abstraction of Secret
Quorums, which extends cryptographic sortition to the selection of an anonymous
quorum of expected size $k$. A node can prove that it is legitimately involved
in a quorum by showing a proof that does not reveal its identity.

\item \textbf{Implementation via rVRF.} We construct Secret Quorums using
ring verifiable random functions in an asynchronous network, providing a first practical mechanism to
realize them in a decentralized setting.

\item \textbf{\protocolname{} Protocol.} We design \emph{\protocolname}, a
protocol for fractional spending based on Secret Quorums. Compared with the
previous protocol~\cite{breakingFp1}, \protocolname{} reduces payment latency
from five to three communication steps and improves settlement complexity
from $O(n^3)$ to $O(n^2)$, while tolerating a rushing-adaptive adversary.
\end{itemize}

The rest of the paper is organized as follows.
Sections~\ref{sec:relatedwork} and~\ref{sec:model} present related work and the
system model, respectively.
We detail the Secret Quorums abstraction in Section~\ref{sec:sq} and present an
implementation based on ring verifiable random functions in
Section~\ref{sec:implem}.
We present \protocolname, a protocol resulting from the application of Secret
Quorums to the fractional spending problem, in Section~\ref{sec:application}.
Section~\ref{sec:conclusion} concludes the paper and discusses future work.

\section{Related Work}\label{sec:relatedwork}

\paratitle{Quorums and cryptocurrencies}
Quorum systems were first introduced in the context of distributed databases~\cite{majorityConsensus,weightedVoting} to ensure consistency and availability.
Quorums were then formalized in the context of byzantine fault tolerance~\cite{byzQuorumSystems} and used in many protocols in the context of consensus~\cite{sokConsensus}, such as PBFT~\cite{pbft}, Tendermint~\cite{tendermint}, HotStuff~\cite{hotstuff}, Ethereum~\cite{gasper}, Morpheus~\cite{morpheusConsensus} and many others.

Cryptocurrencies were typically built on top of consensus protocols~\cite{sokConsensus} (such as Bitcoin~\cite{btcBackbone,btc} and Ethereum~\cite{gasper}), which are used to totally order transactions in a log (or blockchain) and ensure agreement on a common transaction history.
Seminal works~\cite{consensusNumberCryptocurrency,nonConsensus} showed that quorums are enough to validate transactions without the need of a total order, and thus without reaching consensus.
Thus, many recent works~\cite{fastpay,zef,sui,astro,consensusNumberCryptocurrency,nonConsensus,posnoconsensus,stingray,cryptoconcurrency}, such as Stingray, Astro or Fastpay to name a few, rely on quorums to design faster protocols for cryptocurrencies.
All these protocols validate transactions with quorums of size $2f+1$, where $f$ is a bound on the number of byzantine nodes. 

The use of smaller quorums became feasible after quorum systems were extended to a probabilistic setting~\cite{probaQuorums}.
Probabilistic quorums empowered the design of many protocols, such as Algorand~\cite{algorand,algorandScalingByzAgree, asynchronousAlgorand} and ProBFT~\cite{probft}, which use small quorums as validation committee for proposed blocks; protocols for Byzantine Reliable Broadcast~\cite{scalablebyzantinereliablebroadcast, dynamicprobabilisticreliablebroadcast} which uses small quorums to ensure a single message is delivered; and the fractional spending protocol of~\cite{breakingFp1}, which uses small quorums to validate non-conflicting transactions in parallel.

\paratitle{Secret committee election}
To mitigate adaptive corruption, many modern protocols rely on cryptographic sortition or secret committee election.
Algorand~\cite{algorand,algorandScalingByzAgree,asynchronousAlgorand} introduced private and non-interactive committee selection using verifiable random functions (VRFs): validators locally determine whether they are selected and later reveal a proof of selection.
Several subsequent works adopted similar approaches for committee sampling~\cite{faitAccompli}, or sharding~\cite{Yggdrasil}.
These approaches protect selected participants before they speak (in the so-called \say{You Only Speak Once} paradigm), but selected validators eventually reveal themselves when broadcasting their messages or proofs.
Secret Single Leader Election (SSLE)~\cite{ssle,sassafras} further strengthens this idea by hiding the elected leader until it acts.
However, these protocols focus on selecting a single leader and do not provide anonymous quorum validation. 



\paratitle{Anonymous authentication primitives}
Ring signatures~\cite{ringsignatures} enable a signer to prove membership in a set of public keys without revealing its identity.
They are notably used in CryptoNote~\cite{cryptonotes} and deployed in cryptocurrencies such as Monero to provide anonymity to the issuer of a payment.
Our work uses ring signatures for a different purpose: protecting validators participating in quorum executions rather than hiding transaction senders.

Our construction relies on ring verifiable random functions (rVRFs)~\cite{rVRF}, which combine verifiable random functions with ring-signature anonymity.
The same primitive was previously used in Sassafras~\cite{sassafras} to implement secret single leader election, assuming a synchronous network.
In contrast, we integrate rVRFs into a protocol designed for anonymous quorum validation in asynchronous networks with adaptive adversaries.


\paratitle{Adaptive security} 
Adaptive security has traditionally been achieved through proactive secret sharing, threshold cryptography, or delayed revelation mechanisms. 
These techniques are used in settings such as secure broadcast~\cite{adaptivelySecureBroadcast,adaptivelySecureBroadcast2,scalablebyzantinereliablebroadcast, dynamicprobabilisticreliablebroadcast}, multiparty computation~\cite{adaptivelySecureMPC}, zero-knowledge proofs~\cite{adaptiveZeroKnowledge}, random beacons~\cite{adaptivelySecureRandomBeacons}, transaction order fairness~\cite{universallyComposableTransactionOrderFairness}, sharding~\cite{sokSharding} and fractional spending~\cite{breakingFp1}.
In the context of randomized consensus, \cite{adaptiveBA} proposes protocols resilient to an adaptive adversary by relying on a new binding property, which requires hiding the decided values until an additional round of communication, but without hiding the identities of the participating nodes (all nodes participate).
Our work is complementary: instead of protecting already-identified participants from corruption, Secret Quorums prevent the adversary from learning which validators participated in the first place.

\section{Model}\label{sec:model}
\paratitle{Byzantine Asynchronous Message Passing}
Let $\clients$ be an unbounded set of \emph{clients} and $\Pi$ a finite set of
$n$ \emph{validators} responsible for executing clients' queries. We use
\emph{node} to refer to either a client or a validator. Nodes may be subject to
\emph{Byzantine faults}, meaning that they can arbitrarily deviate from the
protocol; such nodes are called \emph{corrupted}, and the others \emph{correct}.
An unbounded number of clients may be corrupted, while at most $f$ validators
are corrupted during an execution. Nodes communicate by message passing over an
\emph{asynchronous} network, where message delays are unbounded.

\paratitle{Cryptographic Primitives} 
Each node $i$ in the system is equipped with a pair of public/private keys and can be identified by its public key $P_i$. 
Nodes are also equipped with a collision resistant hash function.
We denote by $\textit{Hash}(x||y)$ the hash of the concatenation of arbitrary values $x$ and $y$.

\paratitle{Adversary}
We assume a strongly adaptive rushing adversary that can corrupt up to $f$
nodes. The adversary obtains read and
write access to memory of corrupted nodes,  corrupted nodes may collude, and they remain
corrupted for the rest of the execution.
The adversary is strongly adaptive in the sense that it can instantly corrupt
any node of its choice, still subject to the bound of at most $f$ corrupted
nodes, based on messages and internal states observed during the execution.

The adversary controls message delays and can therefore determine the order in
which messages are delivered. However, it cannot permanently prevent the
delivery of messages exchanged between correct nodes. Since the adversary is
rushing, it can schedule message delivery so that a message is delivered to a
corrupted node before it is delivered to any correct node. In this way, the
adversary may observe the content of messages received by corrupted nodes before
deciding its next actions. For instance, the adversary may observe a message delivered to a corrupted node,
learn information from it, and immediately corrupt another node that is expected
to react to that message. After corrupting that node, the adversary can modify
its local state and prevent it from performing the corresponding reaction.

We also assume private channels. That is, the adversary cannot read the internal
state of correct nodes, nor the content of messages exchanged between correct
nodes. However, the adversary may still learn the source, destination, and size
of each message. To limit traffic-analysis attacks, we assume that many protocol
instances are batched together and communicated through gossip, so that the
adversary cannot infer which nodes participate in which instance.


We assume a computationally bounded adversary that cannot break the underlying
cryptographic primitives, find a private key given the matching public key of a correct node nor forge a signature.

\section{Secret Quorums}\label{sec:sq}
\paratitle{Abstraction}
We extend single secret leader election~\cite{ssle} and secret committee election~\cite{algorand} to a set where some validators are secretly selected to execute a query, or \textit{quorum operation}, and then share a proof that the query was indeed executed without revealing anything about the set's composition.
We model secret quorums as an abstraction $\textit{sq}$ with the following operations: 
\begin{itemize}
    \item $\textit{query}(q, k) \mapsto \mathcal{P}$ can be called by a client to execute a query $q$ by a set $S$ of $k$ validators. 
    It returns a validation proof $P$. 
    \item $\textit{verify}(P,q) \mapsto \{ true, false \}$ can be called by any node and returns true if and only if $P$ is a correct validation proof for the query $q$.
\end{itemize}
Secret Quorums ensure the following properties: 
\begin{itemize}
    \item \textbf{Unpredictability}: for any correct validator $v\in\Pi$, the adversary cannot guess if $v\in S$ with probability greater than $|S|/(n-f)$.
    \item \textbf{Anonymity}: validators stay anonymous among $\Pi$, even after $P$ is verified through $\textit{verify}(P,q)$.
    \item \textbf{Fairness}: The probability for any validator to be in $S$ is uniform.
    In particular, the adversary cannot force a specific subset to be in $S$.
    \item \textbf{k-Cardinality}: $k$ validators are selected to execute the query, i.e., $|S| = k$.
    \item \textbf{Verifiability}: $\textit{verify}(P,q)$ outputs $\textit{true}$ if and only if $P$ was correctly generated by $S$ for the query $q$.
\end{itemize}


\paratitle{Why adaptive corruption fails} 
We schematically represent the difference between standard small quorum-based protocol and Secret Quorums in Figure~\ref{fig:secretquorumoverview}.
Thanks to the unpredictability and fairness properties of Secret Quorums, the adversary, without the proof, can only guess which nodes are in the secret quorum with a probability of $k/(n-f)$ for each node.
If the adversary receives the proof, it cannot learn which nodes are in the secret quorum thanks to the anonymity property.
Thus, the adversary, despite being adaptive, learns nothing that could be useful to break the protocol.
The adversary's best strategy to learn if a node is in the secret quorum is to corrupt nodes randomly, which is not better than a static adversary that would have to choose which nodes to corrupt before the execution starts, without any additional information.

\begin{figure*}[t]
\centerline{
\input{classicalQuorumsvsSecretQuorums.tex}
}
\caption{
Comparison between standard quorum validation and Secret Quorums under adaptive and rushing corruption.
Left: validators explicitly sign validation messages, revealing the quorum membership to the adversary, which uses this information to target specific nodes.
Right: validators are anonymously selected within a candidate set and produce anonymous proofs, hiding the effective quorum inside an anonymity set (all candidate validators).
}
\label{fig:secretquorumoverview}
\end{figure*}

\section{An rVRF-based Construction of Secret Quorums}\label{sec:implem}
This section presents a concrete implementation of the Secret Quorum abstraction using ring verifiable random functions (rVRFs)~\cite{rVRF}. 
The construction combines deterministic candidate selection with probabilistic self-election~\cite{algorand}: for each query, a deterministic anonymity set of candidate validators is derived from the queried data, and each candidate privately determines whether it belongs to the effective quorum by evaluating an rVRF. 
The resulting validation proofs, similar to ring signatures, are publicly verifiable but anonymous, satisfying the secrecy and verifiability requirements of Secret Quorums.

\subsection{Ring verifiable random function}\label{sec:rvrf}
A \emph{ring verifiable random function (rVRF)}~\cite{rVRF} is a cryptographic
primitive that combines a verifiable random function with ring signatures.
It was introduced in~\cite{rVRF}\footnote{\cite{rVRF} was peer-reviewed and
presented at \href{https://simula-uib.com/arcticcrypt2025/}{ArcticCrypt 2025},
which has no formal proceedings. Ring verifiable random functions are also used
in~\cite{sassafras}, which was peer-reviewed and published in proceedings. A
Rust implementation of the scheme is available
\href{https://github.com/davxy/ark-vrf}{here}.}.

A standard verifiable random function (VRF)~\cite{vrf} allows anyone to verify
that a pseudorandom output was correctly computed by a node, rather than chosen
arbitrarily. This verification is usually performed with respect to the public
key corresponding to the node's private key. As a consequence, the proof also
identifies the node that produced the output, similarly to a digital signature.

A ring verifiable random function provides the same verifiability while hiding
the prover's identity. More precisely, it allows anyone to verify that the output
was correctly produced by some member of a given set of public keys---the
\emph{ring}---without revealing which member produced it.

To implement secret quorums, we assume every node can use a ring verifiable random function, denoted $\textit{rVRF}$, with the following interface:
\begin{itemize}
    
    \item $\textit{Eval}$: Given a private key $S_i$ and a seed $s$, $\textit{Eval}(s, S_i) \mapsto r$ outputs a pseudorandom number $r \in \Nat$. This is used to compute the verifiable random function part of the scheme and $1\leq r \leq \textit{MAX}$. $\textit{Eval}$ satisfies the properties of \emph{uniqueness} ($r$ is unique given $s$ and $S_i$) and \emph{randomness} ($r$ is pseudorandom and independent from $S_i$ and $s$).
    
    \item $\textit{Sign}$: Given a ring $R = \{ P_1, P_2, \ldots, P_{|R|}\}$ of public keys, a private key $S_i$, a seed $s$ and data $d$, $\textit{Sign}(s, d, R, S_i) \mapsto \sigma$ outputs a signature $\sigma$ that can be used to prove that $d$ was signed by one of the members of $R$ (ring signature) but also to recover the result of $\textit{Eval}(s, S_i)$. $\textit{Sign}$ satisfies the property of \emph{anonymity} ($\sigma$ does not allow to find the public key corresponding to $S_i$ otherwise than with an uniform probability on $R$).
    
    \item $\textit{Verify}$: Given a ring $R = \{ P_1, P_2, \ldots, P_{|R|}\}$ of public keys, data $d$, a seed $s$ and a signature $\sigma$, $\textit{Verify}(\sigma, s, d, R) \mapsto (v, r)$ outputs $(v, r)$ where $v = 1$ if and only if $\sigma$ is a valid ring signature for $d$ and the signer of $d$ generated $r = \textit{Eval}(s, S_i)$, with its private key $S_i$ and the seed $s$ (property of \emph{correctness}). $\textit{Verify}$ also satisfies the property of \emph{unforgeability} (without knowing any private key $S_i$ matching any public key of a ring $R$, it is impossible to forge (i.e., produce) a signature $\sigma'$ such that $\textit{Verify}(\sigma', s, d, R)$ outputs $(1, r)$ for any seed $s$ and data $d$).
\end{itemize}

An implementation of such ring verifiable random function is described in~\cite{rVRF}.
It produces a signature in $O(\log n)$ time the first execution, for a ring $R$ of public keys of size $n$, and the signature is of constant size.
If a same signer uses $R$ again, then the construction allows to \say{re-generate} a signature in $O(1)$ time for a different message and seed. 
The verification time is always in $O(1)$.


\subsection{Secret Quorum Construction} 

The construction is summarized by the following algorithms, and the details of the candidate selection are given in Appendix~\ref{sec-a:solution}.

\paratitle{Parameters}
Our implementation is parameterized by $V_k$, denoting the number of candidates validators such that $k \leq V_k \leq n$.
Intuitively, $k$ determines how many validators are required to execute a query, while $V_k$ determines how many validators an adversary must consider as potential participant: the larger $V_k$ is, the more validators are potential participants and thus the more difficult it is for the adversary to guess which ones are actually selected.

We set $\textit{MAX}$, the upper bound on the output of the VRF part of $\textit{rVRF}$, to be orders of magnitude greater than $V_k$: $\textit{MAX} \gg V_k$.
We denote by $\eta$ a constant, such that $0 < \eta < 1$, used for probabilistic computations of the size of a quorum.

\paratitle{Query protocol}
A client initiates a query by contacting all candidate validators, while each candidate locally determines whether it belongs to the effective quorum and, if selected, returns an anonymous validation proof.

A correct client follows Algorithm \ref{alg:validateclientmain} to execute a query $q$.
The client computes a set of $V_k$ validators called \emph{candidates} using a deterministic function of $q$, depicted in Appendix~\ref{sec-a:solution}, Algorithm~\ref{alg:valSelect}.
All \textit{candidates} test if they are in the quorum executing \textit{q} using Algorithm \ref{alg:vrfSelection}.
They produce a number \textit{out} using $\textit{rVRF}.\textit{Eval}$ that is uniformly pseudorandom and unique given a $\textit{seed}$ and their private key $S_i$.
The construction is probabilistic: each candidate independently self-selects with probability $k'/V_k$ (if \textit{out} is lower than $k'/V_k$).
As a result, the effective quorum size is not fixed but follows a binomial distribution centered around $k'$. 
We deliberately over-sample by setting $k' = k/(1-\eta)$ so that the probability of selecting fewer than $k$ validators is negligible.
Thus, this implementation is an approximation of Secret Quorums with high probability, and the parameter $\eta$ allows to tune the trade-off between the probability of selecting fewer than $k$ validators and the expected size of the effective quorum.

The client waits for enough correct ring signatures from selected validators, depending on the specific quorum system.
Note that this number usually differs from $k$ to take into account that some corrupted validators may be selected but not respond.
For instance, a client that wants to wait for the classical quorum of $2f+1$ validators, where $f$ is a bound on the number of byzantine nodes, can use $k = 3f+1$.
For smaller quorums, the client can use any value of $k$ and wait for a number of signatures that is computed based on the expected number of selected corrupted validators, which is $k' * f / n$.

The validation proof is the ring signatures collected with the pseudorandom number $N$ used as seed.
Because a client contacts all \textit{candidates}, the message complexity of $\textit{query}$ is $O(V_k)$.

\begin{algorithm}
    \caption{Query protocol executed by a client}
    \label{alg:validateclientmain}
    
    \begin{algorithmic}[1]
    \Procedure{$\textit{Query}_{client}$}{$q, k$}
        \State $N \gets random()$
        \State $\textit{candidates} := \textit{SelectCandidates}(\textit{q}, V_k)$
        \State send $q,N,k$ to $\textit{candidates}$
        \State wait received enough signatures to form a quorum
        \State $\textit{signatures} := $ received ring signatures
        \If{$\textit{verify}(\{\textit{signatures},N\}, q, k)$}
            \State \Return $\{\textit{signatures}, N\}$
        \Else
            \State \Return $\bot$
        \EndIf
        \EndProcedure
    \end{algorithmic}
\end{algorithm}

\begin{algorithm}
    \caption{Query protocol executed by validators}
    \label{alg:validatevalidators}
    
    \begin{algorithmic}[1]
    \Procedure{$\textit{Query}_{validator}$}{$q,N,k$} \Comment{upon reception of $q,N,k$ from a client $c$}
        \If{$\textit{SelectedByVRF}(q, N, k)$}
            \If{$q$ is a valid query for this validator}
                \State $r := \textit{GenerateRingSignature}(q, N, k)$
                \State update local state to reflect the validation of $q$
                \State send (anonymously via gossip) $r$ to $\textit{c}$
            \EndIf
        \EndIf
        \EndProcedure
    \end{algorithmic}
\end{algorithm}

\begin{algorithm}
    \caption{VRF based Selection}
    \label{alg:vrfSelection}
    
    \begin{algorithmic}[1]
    \Procedure{\textit{SelectedByVRF}}{$\textit{q}, N, k$}
        \If{$\textit{pkv} \in \text{SelectCandidates}(q, V_k)$}
            \State $\textit{seed} := \textit{Hash}(\textit{q} || N)$
            \State $\textit{out} := \textit{rVRF}.\textit{Eval}(\textit{seed}, S_i)$ \Comment{$1 \leq \textit{out} \leq \textit{MAX}$}
            \State $\textit{threshold} := k' * \textit{MAX} / V_k$
            \State \Return $\textit{out} \leq \textit{threshold}$
        \Else
            \State\Return $\textit{false}$
        \EndIf
    \EndProcedure
    \end{algorithmic}
\end{algorithm}

\begin{algorithm}
    \caption{Ring signature as validation for a transaction}
    \label{alg:ringgen}
    
    \begin{algorithmic}[1]
    \Procedure{\textit{GenerateRingSignature}}{$\textit{q}, N, k$}
        \State $\textit{seed} := \textit{Hash}(\textit{q} || N)$
        \State $\textit{ring} := \text{SelectCandidates}(\textit{q}, V_k)$
        \State \Return $\textit{rVRF}.\textit{Sign}(\textit{seed}, \textit{q}, \textit{ring}, S_i)$
    \EndProcedure
    \end{algorithmic}
\end{algorithm}

\paratitle{Verify Protocol}
We use ring signatures as proofs of validation, which---unlike classical digital signatures---hide the specific signer.
A ring signature only proves the signer was indeed part of the quorum thanks to the VRF.
Algorithm \ref{alg:ringgen} shows that such signature uses the whole set \textit{candidates} as a ring.
This ensures that a valid signature can only be produced by a validator which is in the set \textit{candidates}, without requiring any additional validation beside the signature itself.
Such signature also allows to retrieve the result of $\textit{rVRF}.\textit{Eval}$ by the secret signer.

Nodes can verify if a validation proof for a query $q$ is correct using Algorithm \ref{alg:ringverif}.
It recomputes the \textit{seed}, the \textit{ring} and the \textit{threshold} from \textit{q} at lines \ref{code:verifseed} to \ref{code:verifthreshold}.
They are used to verify that each ring signature \textit{r} in \textit{signatures} is valid line \ref{code:verifcheck}, i.e., $v=\textit{true}$ if and only if was produced by $\textit{rVRF}.\textit{Sign}$ using $\textit{seed}$, $\textit{q}$ and $\textit{ring}$; and $o$ is the result of $\textit{rVRF}.\textit{Eval}$ by the secret signer using the same parameters.
It also checks if a single validator did not produce many ring signatures by keeping track of the pseudorandom output from the VRF.
Indeed, a validator has only one output from the VRF and selected validators will have different pseudorandom number with high probability.
Ring signatures passing this check are collected in a set \textit{valids}, which constitute a proof of validation for $q$ if and only if it contains enough signatures to form a quorum, depending on the specific quorum system.

\begin{algorithm}
    \caption{Verify with ring proofs}
    \label{alg:ringverif}
    
    \begin{algorithmic}[1]
    \Procedure{\textit{verify}}{$P=\{\textit{signatures},N\}, \textit{q}, k$}
        \State $\textit{seed} := \textit{Hash}(q || N)$\label{code:verifseed}
        \State $\textit{ring} := \text{SelectCandidates}(q, V_k)$
        \State $\textit{threshold} := k' * \textit{MAX} / V_k$\label{code:verifthreshold}
        \State $\textit{valids} := \emptyset$
        \State $\textit{outs} := \emptyset$
        \For{$r \in \textit{signatures}$}
            \State $(\textit{v}, o) := \textit{rVRF}.\textit{Verify}(r, \textit{seed}, \textit{q}, \textit{ring})$
            \If{$\textit{v} = \textit{true} \land o \leq \textit{threshold} \land o\not\in \textit{outs}$}\label{code:verifcheck}
                \State $\textit{valids} := \textit{valids} \cup \{ r \}$
                \State $\textit{outs} := \textit{outs} \cup \{ o \}$
            \EndIf
        \EndFor
        \State \Return $|\textit{valids}|$ is enough to form a quorum
    \EndProcedure
    \end{algorithmic}
\end{algorithm}

\paratitle{Communication assumptions} 
Although ring signatures hide the signer's identity cryptographically, network
metadata may still leak the sender's identity. 
To mitigate this, validators
relay validation proofs through gossip. 
As in~\cite{breakingFp1}, we further assume that multiple queries are batched, that non-selected candidate validators also reply to the client (with empty proofs, which does not defeat the purpose of small quorums since the \say{tail} latency still depends on the effective quorum size), and that all messages are sent through private channels, so that an adversary observing network traffic cannot distinguish which validators correspond to which query.

\paratitle{Proofs}
\begin{lemma}[Fairness/Cardinality: selection is uniformly random]\label{lemma:vrfrandom}
    The VRF based quorum selection of Algorithm~\ref{alg:vrfSelection} selects $k$ validators with overwhelming probability (for the security parameter $\eta$) and they are uniformly pseudorandom if $\textit{N}$ is pseudorandom.
\end{lemma}
\begin{proof}[Proof (Sketch)]
    A validator is selected to be part of the quorum by Algorithm \ref{alg:vrfSelection} if and only if $\textit{out}$, the result of $\textit{rVRF}.\textit{Eval}(\textit{seed}, S_i)$ is lower than a threshold $\textit{threshold} = k' * \textit{MAX} / V_k$ where $V_k$ is the number of candidate validators and $k' = k / (1 - \eta)$.
    Because of the randomness property of the Ring VRF scheme, $\textit{out}$ is uniformly pseudorandom and $1 \leq \textit{out} \leq \textit{MAX}$. 
    Moreover, $\textit{MAX} \gg V_k$ gives us $k'$ random validators selected in expectation. 
    
    Let $X$ be the random variable counting the number of validators that are selected and we look for the probability that less than $(1-\eta)k' = k$ validators are selected. 
    The probability that a validator is selected is $k'/V_k$ and $X$ follows a binomial distribution $X \sim Binomial(V_k, k'/V_k)$. 
    Therefore using Chernoff's lower tail bound \cite{chernoff}:
     \begin{equation*}
        \begin{aligned}
        P(X \leq k) & = P(X \leq (1-\eta)k') \\ 
          & \leq \exp^{-\frac{\eta^2k'}{2}} \\
          & = \exp^{-\frac{\eta^2k}{2(1-\eta)}}
        \end{aligned}
        \end{equation*}
    which is a negligible function of $k$ for $k$ large enough.
    Thus, at least $k$ random validators are selected with high probability.
\end{proof}

\begin{lemma}[Fairness: selection cannot be manipulated]\label{lemma:vrfnomanipulation}
    The adversary cannot force a specific subset of validators to be selected by the VRF in Algorithm~\ref{alg:vrfSelection} for $k$ large enough. 
\end{lemma}

\begin{proof}[Proof (Sketch)]
    The VRF seed used in Algorithm~\ref{alg:vrfSelection} depends on $N$, which is sent by the client in Algorithm~\ref{alg:validateclientmain}.
    If the client is correct, then $N$ is uniformly pseudorandom and so is the selection (Lemma \ref{lemma:vrfrandom}). 
    If the client is corrupted, the adversary can try to reshuffle $N$ (a grinding attack) until enough corrupted validators are selected. 
    However, the expected number of corrupted validators selected is $k' * f / V_k$ for each choice of $N$.
    For the adversary to break safety, they typically need to corrupt a constant fraction of the quorum, say $t = \alpha k'$ for some $\alpha > f/V_k$. 
    The probability that a single random N yields at least $t$ corrupted validators selected is given by the tail of a binomial distribution.
    By a Chernoff upper tail bound~\cite{chernoff}, finding a suitable $N$ thus requires an exponential number of attempts in $k' * f / V_k$ (as also detailed in~\cite{breakingFp1}) and the adversary is computationally bounded. 
\end{proof}

\begin{lemma}[Unpredictability: selection is secret]\label{lemma:vrfsecret}
   A correct validator is the only one able to know if it is selected in the quorum.
\end{lemma}

\begin{proof}[Proof (Sketch)]
    One needs to run $\textit{rVRF}.\textit{Eval}$ to know if a validator has been selected by Algorithm \ref{alg:vrfSelection}.
    It requires to know $S_i$ the private key of the validator, which is supposed to remain secret (if the validator is correct).
\end{proof}

\begin{lemma}[Anonymity: Ring Signatures hide identities]\label{lemma:ringsnoleak}
    If the adversary learns a ring signature, it is not able to determine which validators signed it. 
\end{lemma}

\begin{proof}[Proof (Sketch)]
    This follows directly from the anonymity property of $\textit{rVRF}.\textit{Sign}$ applied to the set of candidate validators (minus corrupted but selected validators).
\end{proof}

\begin{lemma}[Verifiability: Correctness of Ring Signatures]\label{lemma:ringcorrectness}
    A set of ring signatures is a valid proof (returns $\textit{true}$ in Algorithm \ref{alg:ringverif}) only if it was produced by a quorum.
    Moreover, a quorum can produce a valid set of ring signatures.
\end{lemma}

\begin{proof}[Proof (Sketch)]
    Let \textit{q} be an invalid query.
    Because no one can know who produced a ring signature (Lemma \ref{lemma:ringsnoleak}), it could be possible for a corrupted validator $V_k$ selected by Algorithm \ref{alg:vrfSelection} to produce multiple signatures such that even if every correct validator selected for \textit{q} does not produce any validation, there is still enough duplicate validations from $V_k$ to form a valid proof. 
    However, due to the uniqueness property of the Ring Signatures scheme, there is a unique value $\textit{out}$ such that for a seed $s$ and private key $S_i$, $\textit{Eval}(s, S_i) = \textit{out}$ and for any signature $\sigma$ such that $\textit{Verify}(\sigma, s, m , R) = (1, \textit{out}')$, $\textit{out}' = \textit{out}$. 
    In other words, all duplicates from a same validator $V_k$ need to share the same value \textit{out}.
    Because we remove ring signatures with same $\textit{out}$ value in Algorithm \ref{alg:ringverif}, $V_k$ cannot produce multiple signature to influence the validation of \textit{q}. 
    Conversely, one could argue that it would also prevent correct validators from producing valid proofs.
    However, having 2 correct validators with same $\textit{out}$ value only happens with negligible probability because $\textit{MAX} \gg V_k$.
\end{proof}

\begin{theorem}
The rVRF-based construction implements the Secret Quorum abstraction with high probability, assuming the correctness, uniqueness, and anonymity properties of the underlying rVRF scheme.
\end{theorem}
\begin{proof}
    Unpredictability follows from in Lemmas \ref{lemma:vrfrandom}, \ref{lemma:vrfnomanipulation} and \ref{lemma:vrfsecret}.
    Anonymity follows from Lemma~\ref{lemma:ringsnoleak}, with anonymity
    set bounded by $V_k$ (maximal when $V_k=n$).
    The fairness property follows from Lemmas \ref{lemma:vrfrandom} and \ref{lemma:vrfnomanipulation}.
    The k-cardinality property is proved in Lemma \ref{lemma:vrfrandom} with high probability.
    Finally, the verifiability property is shown in Lemma \ref{lemma:ringcorrectness}.
\end{proof}

\subsection{Discussion and Limitations}\label{sec:limits}
Our construction achieves secrecy and verifiability via randomized self-selection and ring-based proofs, but it also introduces trade-offs and limitations that should be carefully considered in practice.

\paratitle{Overhead versus standard quorum systems}
Compared to traditional quorum systems~\cite{byzQuorumSystems,probaQuorums}, our implementation replaces identity-revealing validation by anonymous validation over a candidate set of size $V_k$. 
This introduces an additional per-query cost of contacting all $V_k$ candidates and running rVRF-based selection and verification, whereas a classical quorum only needs to communicate with the selected validators.
In exchange, the quorum membership remains hidden after validation, preventing adaptive adversaries from targeting validators and thus avoiding expensive protection mechanisms that would be necessary if membership was revealed by the proof (such as secret-sharing~\cite{breakingFp1} for instance).
Thus, the overhead is shifted from expensive mechanism to share identity-revealing proofs to a larger but one-shot validation cost.
In practice, this trade-off is favorable when $V_k$ is chosen moderately above $m$.
Regarding latency, both standard quorum systems~\cite{byzQuorumSystems,probaQuorums} and our approach require $2$ message delays to execute a quorum query.

\paratitle{Probabilistic approximation}
Selection is probabilistic: the effective quorum is sampled from the candidate set and the protocol guarantees that at least $k$ validators are selected only with high probability (see Lemma~\ref{lemma:vrfrandom}). 
In practice one sets $k' = k/(1-\eta)$ and tunes $\eta$ to drive the Chernoff tail bound below an acceptable failure probability. 
Designers should pick parameters (in particular $k',V_k,\eta$) to meet a concrete failure target rather than relying on asymptotic statements alone.

\paratitle{Anonymity depends on $V_k$}
Anonymity is provided with respect to the candidate ring of size $V_k$. 
Larger $V_k$ increases the anonymity set and reduces the advantage of an adaptive adversary, at the cost of higher per-query communication (clients contact all $V_k$ candidates). 
The strongest anonymity claim holds when $V_k=n$; and the anonymity guarantees degrade proportionally to $V_k$.

\paratitle{Network metadata assumptions}
The construction assumes a sufficiently obfuscated dissemination mechanism. 
In practice, traffic analysis tools~\cite{moneroNetwork, lightningCensorship, deanonymizingeth, studyDoS}, timing correlation, or an adversary that controls delivery order can reduce anonymity even if cryptographic proofs remain sound. 
Practical deployments should consider gossip, batching, dummy traffic, or network-layer protections as additional protections to mitigate metadata leakage.

\paratitle{Requires an rVRF primitive}
Our implementation depends on a ring verifiable random function that provides uniqueness, randomness, unforgeability and signer-anonymity. If a secure, efficient rVRF implementation is unavailable, the construction cannot be instantiated as presented. 
Any concrete deployment would critically depend on the implementation of this cryptographic primitive.

\paratitle{Candidate collisions and \textit{MAX}}
The verification step rejects duplicate VRF outputs to prevent a single corrupted signer from inflating the counted proofs. 
Two distinct correct validators may nevertheless collide on the VRF output if \textit{MAX} is not large enough. 
The collision probability is about $\binom{V_k}{2}/\textit{MAX}$, so we require $\textit{MAX} \gg V_k$ (for instance, \textit{MAX} on the order of $2^{\lambda}$ for security parameter $\lambda$) to make collisions negligible. 
If collisions become non-negligible the protocol either undercounts valid signers or must accept a higher false-negative rate.

\paratitle{Access Strategy} Our construction only supports a specific access strategy, where the validators selected for a query are randomly sampled.
This is sufficient for many applications, including fractional spending, but it does not support arbitrary access strategies where the quorum membership is determined by more complex functions of the query or system state, such as involvements in previous quorums~\cite{k-quorums}.
However, it is unclear whether a more general access strategy can be secure against adaptive adversaries without revealing some information about the quorum membership, which would defeat the purpose of Secret Quorums.

These limitations reflect trade-offs between anonymity, communication, and practical cryptographic assumptions. 
They do not invalidate the abstraction but should guide parameter selection and deployment choices.

\section{Application to Fractional Spending}\label{sec:application}
\subsection{Overview}
Fractional Spending is a recent approach to transaction validation in
asynchronous Byzantine systems, introduced by Bazzi and Tucci~\cite{breakingFp1}.
It enables non-conflicting transactions to be validated in parallel, using quorums whose size may be smaller than the corruption threshold $f$. 
This is achieved through a new family of quorum systems, called $(k_1, k_2)$-quorums, guaranteeing that up to $k_1$ transactions can be processed concurrently while preventing more than $k_2$ concurrent transactions from being validated.

State-of-the-art fractional spending protocols~\cite{breakingFp1} leverage these quorum systems to allow up to $k_1$ parallel payments, each spending roughly $1/k_2$ of an account balance.
Once all spendable fractions of an account have been consumed, a settlement phase contacts validators involved in previous validations to recover the remaining balance or allow sellers to redeem received funds.

Fractional spending is particularly vulnerable to adaptive adversaries because validation quorums are selected dynamically and may contain fewer than $f$ validators. 
Once validators participating in a payment validation are identified, an adaptive adversary may adaptively corrupt them before settlement propagates and change the validation status of the payment.
To address this issue, Bazzi and Tucci~\cite{breakingFp1} rely on blind signatures and secret sharing. 
While effective against adaptive corruption, these mechanisms introduce substantial communication and latency overheads. 
In particular, settlement and redeem operations require multiple instances of secret sharing among validators, resulting in high message complexity.

\protocolname replaces these protections with Secret Quorums. 
Instead of revealing the identities of validators participating in a payment validation, selected validators produce anonymous yet publicly verifiable validation proofs. 
As a result, validators remain hidden even after validation completes,
preventing adversaries from targeting them adaptively.
At a high level, \protocolname preserves the structure of the fractional spending protocol of~\cite{breakingFp1} while replacing identity-revealing quorum validations with Secret Quorum validations. Buyers submit transactions that are validated by secretly selected validators through \texttt{sq.query(tx)}. 
Validators produce anonymous validation proofs that sellers can later use during settlement and redeem operations. 
Because these proofs no longer reveal quorum membership, they can be broadcast directly without relying on expensive secret-sharing mechanisms.

This modification significantly reduces communication complexity. Compared to previous approaches, \protocolname reduces settlement complexity from $O(n^3)$ to $O(n^2)$ and redeem complexity from $O(n^2)$ to $O(n)$ while maintaining resilience against adaptive byzantine adversaries.

\subsection{The Fractional Spending Problem}\label{sec:fractionalproblem}
We now recall the Fractional Spending Problem. 
The interested reader can find the original formulation of the problem in \cite{breakingFp1} and a formulation generalising the  Asset Transfer problem~\cite{consensusNumberCryptocurrency} to fractional spending in Appendix \ref{sec-a:formalization}.

The problem introduces the notion of \emph{funds}\footnote{A fund here corresponds to a \say{fully validated fund} in \cite{breakingFp1}.}.
Akin to \emph{accounts}, funds carry a \emph{balance} of the currency and have some clients called \emph{owners} that can apply operations on them.
Depending on the context, a client can act as a \emph{buyer} or as a \emph{seller}: a buyer is the one spending from a fund, and a seller is the one getting a payment.
Any fund is denoted as $F=\{F.owner, F.balance\}$. Some funds are created at system setup, and new funds are created via a successful settlement. 
We make no assumption on ownership of funds, and thus funds can be shared between multiple clients. 
We do not need Consensus since owners make independent payments of the \textit{same amount}. This can be deduced from~\cite{consensusNumberCryptocurrency} and is discussed in Appendix \ref{sec-a:noconsensus}. 

The fractional spending problem is parameterized by $s_1$, denoting the number of concurrent payments that can be validated in parallel; and by $s_2$, denoting the fraction such that, when issuing a fractional payment, $\frac{1}{s_2}$ of the balance of the fund is spent: in total $\frac{s_1}{s_2}$ of the balance can be spent.
Thus, our problem is well defined only if we consider $s_1 \leq s_2$. 
We should note here that the balance which is part of the fund definition does not change due to spending from the fund but the payment protocols guarantee that no overdraft occurs. The new balance is calculated as part of settlement when a new fund with a new balance reflecting the payments is created.

Any client (buyer or seller) can invoke the following three operations: \texttt{Pay(F,$pk_s$)}, \texttt{Settlement(F)}, \texttt{Reedem(F)}:
\begin{itemize}
    \item \texttt{Pay(F,$pk_s$)} is used  by any buyer to execute a fractional spending from its own fund F to the seller identified by $pk_s$. Any spending transfers $\frac{F.balance}{s_2}$. 
    
    \item \texttt{Settlement(F)} is used  by any buyer to create a new fund with unspent balance. 

    \item \texttt{Redeem(F)}\footnote{The operation redeem is also called settlement for the seller in \cite{breakingFp1}. We renamed the operation for clarity but it is only syntactical.} is used  by any seller $pk_s$ to create a new fund that collects all the fractional spending received through a \texttt{Pay(F,$pk_s$)}.
\end{itemize}

The fractional spending problem is specified through the following properties
\begin{enumerate}
    \item \textbf{No double spending:} (Safety) The sum of the value of all payments from a fund cannot exceed its balance.
    
    \item \textbf{Non-interference:} (Liveness) A buyer should be able to make at least $s_1$ payment in parallel successfully.

    \item \textbf{Eventual Retrieving} (Liveness) Any correct buyer or seller invoking, respectively, \texttt{Settlement(F)} and \texttt{Redeem(F)}  will  eventually own a new fund.  
     
    \item \textbf{Retrieving Consistency} (Safety) the balance of new funds created through \texttt{Settlement(F)} and \texttt{Redeem(F)}, respectively, do not exceed the unspent balance and do not exceed the sum of the fractional spending received.  
\end{enumerate}

\subsection{\texorpdfstring{$(k_1, k_2)$}{(k1,k2)}-Quorums and the Previous Fractional Spending Protocol }\label{sec:initialprotocol}
For completeness, we present a high-level description of the fractional spending protocol of Bazzi and Tucci~\cite{breakingFp1}, that \protocolname improves upon.
We start by recalling the notion of $(k_1, k_2)-$quorums, which are used by the protocol of Bazzi and Tucci~\cite{breakingFp1} (and thus also by \protocolname); and then we sketch the protocol itself. 

\paratitle{$(k_1, k_2)-$ quorums: combining intersection and non-intersection}\label{sec:k1k2}
$(k_1, k_2)-$quorum systems from~\cite{breakingFp1} have both: a lower bound on the size of the intersections to check operations can be safely applied; and an upper bound to control concurrency.
$(k_1, k_2)-$quorums ensure with high probability that, for a set of transactions $T$, $k_1$ transactions can be concurrently validated but no more than $k_2$ can:
\begin{itemize}
    \item \textbf{$(k_1, k_2)-$liveness:} $T$ can be entirely validated in parallel if $|T| \leq k_1$.
    \item \textbf{$(k_1, k_2)-$safety:} no more than $k_2$ will be validated if $|T| \geq k_2$.
\end{itemize}
$(k_1, k_2)-$quorum systems are made possible by a carefully crafted random sampling of subset of traditional quorums to ensure liveness; and by remarking that taking the union of enough small subsets of traditional quorums reconstitute a traditional quorum, which ensure safety.
We discuss the fault-tolerance of $(k_1,k_2)-$quorum systems in Appendix~\ref{sec-a:ftk1k2}.

\paratitle{Previous Protocol}
The original protocol proceeds in two phases. 
During the \emph{validation phase}, a buyer submits a transaction to a $(k_1,k_2)-$quorum of validators responsible for checking whether sufficient spendable fractions remain available. 
Validators return signed validation proofs that allow the seller to accept the payment without requiring global consensus.

Once all spendable fractions of an account have been consumed, the protocol enters a \emph{settlement phase}. 
During settlement, validators involved in
previous transaction validations are contacted to reconstruct the residual
state of the account and compute the remaining balance. Similarly, sellers
redeem received payments by collecting validation evidence from validators that
previously approved transactions.

The protocol in \cite{breakingFp1} uses classic digital signatures from validators: a fractional payment is valid if signed by a $(k_1, k_2)-$quorum. 
Settlement and redeem rely on \emph{secret sharing}~\cite{secretSharing}.
Without it, the adversary can learn the composition of the quorum $Q$ used in a payment when its proofs are shared. With this knowledge, the adversary can instantly erase the proofs from all validators in $Q$---because the size of $Q$ is less than $f$---and correct validators would invalidate this redeem, preventing a correct seller from redeeming a payment that was successful.

A high-level presentation of a payment in \cite{breakingFp1} is as follows.
\begin{algobox}{\textbf{Sketch of a payment in previous work~\cite{breakingFp1}}}
\begin{enumerate}
        \item The buyer creates a transaction $tx$ with a $\textit{fund}$ and $seller$, and sends it to the seller.
        \item The seller randomly draws a nonce which determines a random $(k_1, k_2)$-quorum for validation. \label{code:randomnoncebasedqselection}
        \item The seller hashes the identities of the members of the quorum and sends the hashes to the buyer. The buyer sign these hashes to produce a blind signature of the quorum. 
        Blind signatures attest the buyer agreed on this quorum, without revealing its composition. Thus, a corrupted $seller$ cannot contact too many validators for a same transaction, which would prevent the buyer from making as much different transactions as possible later. \label{code:blindsignatures}
        \item The seller contacts the quorum for validation, with the transaction and the blind signatures from the buyer.
        \item After checking if the blind signatures are indeed from the buyer, correct validators receiving a new transaction to validate check if they already validated a transaction from the same fund (for the same transaction but also for a different one). If they did, they do not send any validation. 
        \item If they did not, they produce a digital signature of the transaction with their private key and send it to the seller.
        \item The payment is completed when 2 thirds of the quorum replied with a validation. 
\end{enumerate}
\end{algobox}
Such protocol has a message complexity of $O(|Q|)$, where $Q$ is the validating quorum, and its latency adds up to $5$ message delays.


A high-level description of a settlement is as follows.
\begin{algobox}{\textbf{Sketch of a settlement in previous work~\cite{breakingFp1}}}
\begin{enumerate}
    \item The buyer sends the $\textit{fund}$ to settle to all validators.
    \item Validators send the transaction they validated (if any) from $\textit{fund}$ to each other via secret sharing (because we cannot trust the buyer to exhaustively list its payments). \label{code:secretsharing1}
    \item When validators receive $n-f$ transactions (or null) from others, they reconstruct the set of transactions emitted from $\textit{fund}$.
    \item Each validator creates a new fund, whose balance is the balance of $\textit{fund}$ minus each transaction gathered via secret sharing.
    \item Each validator signs the new fund and sends it to the buyer.
    \item $n-2f$ signatures received by the buyer for a same fund concludes the settlement.
\end{enumerate}
\end{algobox}
Because of secret sharing, message complexity of settlement is $O(n^3)$.

A seller executes the following steps to redeem received payments.
\begin{algobox}{\textbf{Sketch of a redeem in previous work~\cite{breakingFp1}}}
\begin{enumerate}
    \item The seller retrieves from its memory digital signatures from quorums produced during each payment.
    \item The seller sends the fund to redeem along the validations proofs to all validators via secret sharing.\label{code:secretsharing2}
    \item Validators check the signatures and if they are correct, they create a new fund whose balance is the sum of the payments received by the seller.
    \item The new fund is signed and the signature sent to the seller.
    \item The seller collects $n-f$ signatures to conclude the redeem.
\end{enumerate}
\end{algobox}
The message complexity of this redeem is $O(n^2)$ because of secret sharing.

\paratitle{Limitations}\label{sec:limitations}
Producing blind signatures during a payment (step \ref{code:blindsignatures}) require an interactive phase between the buyer and the seller, resulting in a high latency.
Each instance of secret sharing costs $O(n^2)$ messages and $5$ message delays, resulting in a high message complexity for settlement (step \ref{code:secretsharing1}) and redeem (step \ref{code:secretsharing2}).
These mechanisms are protections against an adaptive adversary and we replace them in \protocolname with a more efficient one: Secret Quorums.

\subsection{\protocolname}\label{sec:protocol}
We propose \emph{\protocolname}, a new protocol for fractional spending based on Secret Quorums. 
\protocolname preserves the structure and security guarantees of the protocol of~\cite{breakingFp1} while replacing its identity-revealing quorum validation mechanism with anonymous quorum validation. 
This modification protects small validation quorums against adaptive adversaries, reduces the communication complexity of both settlement and redeem operations, and reduces the latency of payments.

Because \protocolname closely follows the protocol of~\cite{breakingFp1}, we focus on the components modified by Secret Quorums and highlight the key differences. 
For completeness, the full protocol is provided in Appendix~\ref{sec-a:solution}.
We denote by $m$ the size of any $(k_1, k_2)$-quorum (see~\cite{breakingFp1} for how to construct such quorums) and use it as the target size of a query to the Secret Quorum abstraction $\textit{sq}$, so that the selected validators form a valid $(k_1, k_2)$-quorum with high probability.
In the Secret Quorum implementation, validators are selected from a candidate set of size $V_m$, which forms the anonymity set of the protocol. 
The effective quorum is the subset of these $V_m$ candidates that is privately elected by $\textit{sq.query}(\textit{tx}, m)$, with expected size $m$.
Thus, $m$ controls the quorum size required by the fractional spending protocol, while $V_m$ controls the anonymity level: larger values of $V_m$ increase the size of the anonymity set at the cost of additional communication during payment validation.
Secret quorums allow to solve the fractional spending problem with $s_1=k_1$ parallel transactions and each transaction spends $\frac{1}{s_2}=\frac{1}{k_2}$ of the fund's balance.

\paratitle{Using Secret Quorums}
The main difference between \protocolname and the protocol of~\cite{breakingFp1} lies in the validation phase. 
In the original protocol, a seller selects a quorum through a random nonce and interacts with the buyer to obtain blind signatures of the selected validators. 
Moreover, validators produce standard signatures as validation proofs, which are publicly verifiable but reveal the identities of the validators that approved the payment, requiring secret sharing
during settlement and redeem operations to protect them from adaptive
corruption.

In \protocolname, this interactive quorum selection and blind-signature phase is replaced by a single invocation of $\textit{sq.query}(\textit{tx}, m)$, where $\textit{tx}$ is the transaction to validate and $m$ is the size of any $(k_1,k_2)$-quorum. 
Selected validators privately determine their participation and produce an anonymous validation proof. 
Since this proof does not reveal which validators processed the transaction, the identities of the quorum remain hidden even after validation completes.
As a result, validation proofs can be broadcast directly,
eliminating the need for secret sharing.

The correctness of \protocolname follows from the correctness properties of Secret Quorums together with the correctness of the original fractional spending protocol~\cite{breakingFp1}, as \protocolname only replaces the quorum validation mechanism while preserving the remaining protocol structure.
These modifications preserve the logic of the original protocol while changing its communication pattern.

\paratitle{Complexities}
We summarize complexities in Table \ref{tab:complexities}.
Recall that a settlement required each validator to use secret sharing, leading to an $O(n^3)$ message complexity and a redeem required the seller to use secret sharing, leading to an $O(n^2)$ message complexity.
Because Secret Quorums provide anonymous validation proofs, \protocolname eliminates the secret-sharing phase of~\cite{breakingFp1} and removes a factor $n$ in communication complexity: message complexities are only $O(n^2)$ for a settlement and $O(n)$ for a redeem in \protocolname.
Replacing blind signatures with a single call to $\textit{sq}.\textit{query}$ also reduces the latency of a payment from $5$ message delays to $3$ message delays.
The payment phase incurs a cost proportional to $V_m$ rather than $m$, as $\textit{sq.query}$ contacts all $V_m$ candidate validators so that the subset of validators actually involved remains hidden: larger values of $V_m$ increase the anonymity set and make it harder for an adaptive adversary to identify participating validators.
This overhead is the price paid for preserving validator anonymity and eliminating the more expensive secret-sharing phases during settlement and redeem.
In practice, $V_m$ can be chosen independently of $n$ and moderately larger than $m$ to balance these concerns.

Overall, \protocolname illustrates how Secret Quorums can replace identity-revealing small-quorum validations in fractional spending protocols.
By preserving validator anonymity after validation, \protocolname removes the need for costly secret-sharing protections while retaining resilience against adaptive corruption.

\begin{table}[]
    \centering
    \begin{tabular}{|c|c|c|}
        \cline{2-3}
         \multicolumn{1}{c|}{} & Standard Signatures and Secret Sharing \cite{breakingFp1} & \protocolname \\
         \hline
        pay & $O(m)$ & $O(V_m)$ \\  
        settlement & $O(n^3)$ & $O(n^2)$ \\  
        redeem & $O(n^2)$ & $O(n)$ \\  
         \hline
         latency (payment) & 5 message delays & 3 message delays \\
         \hline
    \end{tabular}
    \caption{Communication complexity and payment latency in \cite{breakingFp1} and \protocolname}
    \label{tab:complexities}
\end{table}





\section{Conclusion}
\label{sec:conclusion}
We introduced \emph{Secret Quorums}, an abstraction enabling distributed protocols to tolerate adaptive adversaries by preserving the anonymity of validators participating in quorum operations. 
Unlike classical secret committee or leader election approaches, where committee members are eventually revealed, Secret Quorums maintain anonymity even after quorum operations have completed. 
This prevents adaptive adversaries from identifying and corrupting critical validators before protocol state propagates to correct validators.

We presented an implementation of Secret Quorums based on ring verifiable random functions, combining private self-selection with anonymous yet verifiable validation proofs, without any time overhead compared to classical quorum systems.

We demonstrated the relevance of secret quorums by introducing \emph{\protocolname}, a fractional spending protocol resilient to adaptive adversaries.
\protocolname replaces expensive secret-sharing protection mechanisms with Secret Quorums to significantly reduce communication complexity while maintaining correctness under adaptive Byzantine behavior.

Our approach nevertheless has several limitations. 
Secret Quorums rely on an anonymous gossip assumption, while traffic analysis or timing correlation may weaken validator anonymity in practice.
Moreover, anonymity guarantees depend on the size and distribution of candidate validator sets: if too many validators are corrupted, the anonymity set may become significantly reduced. 
Improving the construction of Secret Quorums to provide closer approximation of effective quorums or more general access strategies is an interesting direction for future work.
While \protocolname improves asymptotic communication complexity compared to previous fractional spending protocols, anonymous gossip dissemination and ring-based cryptographic operations may still introduce non-negligible practical overheads.

Beyond fractional spending, Secret Quorums may be applicable to tolerate adaptive adversaries in committee-based consensus, Byzantine Reliable Broadcast, proof-of-stake protocols, sharding and other distributed systems relying on dynamically selected validators; and several directions remain open for future work. 
First, exploring alternative constructions beyond ring verifiable random functions may improve efficiency or reduce cryptographic assumptions. 
For instance, multi-party computation techniques could be used to compute quorum validation results without revealing validator identities, albeit with potentially higher overheads.
Second, implementing and evaluating Secret Quorums in large-scale distributed settings would provide further insight into their practical throughput, scalability and resilience.


\newpage

\bibliography{refs}

\newpage

\section*{APPENDIX}
\appendix

\section{Formalization of the Fractional Spending Problem}\label{sec-a:formalization}
We generalize the formalization of \cite{consensusNumberCryptocurrency} to the fractional case. 
Let \funds be a set of \emph{funds}, akin to \emph{accounts}.
and $\mu: \funds \mapsto 2^{\clients}$ an \emph{ownership} map that associate each fund with a set of clients allowed to pay from the fund.
We make no assumption on $\mu$ and thus handle funds shared between multiple clients.
Depending on the context, a client can act as a \emph{buyer} or as a \emph{seller}: a buyer is the one spending from a fund, and a seller is the one getting a payment.
The fractional spending problem is parameterized by $s_1$ and $s_2$: $s_1$ denotes the number of fractional payment a buyer can do in parallel;
$s_2$ denotes the fraction such that, when doing a fractional payment, $\frac{1}{s_2}$ of the balance of the fund is spent.
Contrarily to~\cite{breakingFp1}, we batch settlements and redeems of multiple funds by using a set as parameter of these operations.
We define the fractional-spending object type associated with $\funds$ and $\mu$ as a tuple $(Q, q_0, O, R, \Delta)$, where:
\begin{itemize}
    \item The set of possible states is $Q$ the set of all possible maps $q: \funds \mapsto \Nat\times (\clients \mapsto \Nat)$.
    Intuitively, each state of the object assigns each fund with the amount of currency it is holding along with a map registering the number of payments made to each client.
    
    \item The initialization map $q_0: \funds \mapsto \Nat\times(\clients\mapsto 0)$ assigns the initial balance to each fund but without any payment yet.
    
    \item Operations and responses of the type are defined as 
    \begin{equation*}
      \begin{aligned}
        O & = \{pay(F, s) : F\in\funds, s \in \clients\} \\
          & \cup \{settle(S_F): S_F\subseteq\funds\} \\
          & \cup \{redeem(S_F): S_F\subseteq\funds \} \\
          & \cup \{read(F) : F \in \funds\}
      \end{aligned}
    \end{equation*}
    and $R = \{true, false\} \cup \funds \cup (\Nat\times(\clients \mapsto \Nat))$.

    \item $\Delta$ is the set of valid state transitions. For a state $q \in Q$, a client $p \in\clients$ an operation $o \in O$, a response $r \in R$ and a new state $q' \in Q$, the tuple $(q, p, o, q',r) \in \Delta$ if and only if one of the following conditions is satisfied:
    \begin{itemize}
        \item \begin{equation*}\begin{aligned} o & = pay(F, s) \\
          & \land ((r = true \land p \in \mu(F) \land q(F) = (balance, payments) \\
          & \land \sum_{p'\in \clients}(payments[p']) \leq s_2 \text{ (no double spending)} \\
          & \land q'(F) = (balance', payments') \land balance' = balance \\
          & \land payments'[s] = payment[s] + 1 \\
          & \land \forall s' \in \clients \setminus \{ s \}: payments'[s'] = payments[s'] \text{ (no other client received a payment)} \\
          & \land \forall F' \in \funds \setminus \{F\} : q'(F') = q(F') \text{ (all other funds unchanged)}) \\
          & \lor (r = false \land q' = q))
          \end{aligned}
        \end{equation*}

        \item \begin{equation*}
        \begin{aligned}
        o & = settle(S_F)  \\ 
          & \land \forall F \in S_F: p \in \mu(F) \\
          & \land r = F' \land q(F') = (0, no\_payments) \land \forall p'\in \clients: no\_payments[p'] = 0 \land p \in \mu(F') \\
          & \land q'(F') = (balance', payments') \\
          & \land balance' = \sum_{F\in S_F: q(F) = (balance, payments)} (balance - \frac{balance}{s_2}\sum_{p'\in \clients}payments[p']) \\
          & \text{ (subtract amount spent)}\\
          & \land \forall p'\in \clients: payments'[p'] = 0 \text{ (no payment made yet)} \\
          & \land \forall F'\in\funds\setminus S_F: q'(F') = q(F') \text{ (all other funds unchanged)})
        \end{aligned}
        \end{equation*}
        
        \item \begin{equation*}
        \begin{aligned}
            o & = redeem(S_F) \\
              & \land \forall F \in S_F: (q(F) = (balance, payments) \land payments[p] > 0) \\
              & \land r = F' \land q(F') = (0, no\_payments) \land \forall p'\in \clients: no\_payments[p'] = 0  \land p \in \mu(F') \\
              & \land q'(F') = (balance', payments') \\
              & \land balance' = \sum_{F\in S_F: q(F) = (balance, payments)} (payments[p] \frac{balance}{s_2}) \\
              & \text{ (merge amount of payments received)}\\
              & \land \forall p'\in \clients: payments'[p'] = 0 \text{ (no payment made yet)} \\
              & \land \forall F'\in\funds\setminus S_F: q'(F') = q(F') \text{ (all other funds unchanged)})
        \end{aligned}
        \end{equation*}
        
        \item \begin{equation*}
        \begin{aligned}
            o & = read(F) \\
              & \land q = q' \text{ (all funds unchanged)} \\
              & \land r = (balance\_r, payments\_r) \land q(F) = (balance\_q, payments\_q) \\
              & \land balance\_r = balance\_q \text{ (returns the balance)} \\
              & \land p \in \mu(F) \implies payments\_r = payments\_q \text{ (an owner knows all payments)}\\
              & \land p \not\in \mu(F) \implies \forall p'\in\clients\setminus\{p\}: payments\_r[p']=0 \text{ (unknown payments to others)} \\
              & \land payments\_r[p] = payments\_q[p] \text{ (seller knows payments received)}
        \end{aligned}
        \end{equation*}
    \end{itemize}
\end{itemize}

In other words, the operation of fractional payment is denoted $pay(F,s)$ where $F$ is a fund to spend from and $s$ the seller.
Its effect is to increment the payments counter of $F$ such that when the buyer settles a set of funds $S_F$ (in which $F$ is) via $settle(S_F)$, its total value is subtracted from $\frac{1}{s_2}$ times the old balance for every payments recorded.
Similarly, when the seller settles a set of funds $S_F$ via $redeem(S_F)$, his balance is incremented for each payment by $\frac{1}{s_2}$ times the value of the fund the payment was made from. 
Both settle and redeem take effect into a new fresh fund that was not used before. 
Clients can use the $read(F)$ operation to read the state of a fund $F$ but due to the fractional nature of payments, they can only see the payments they received or the one they sent. 
Because the balance is updated only via settle or redeem, they can read the balance of any fund.

This specification is equivalent to the one proposed in \cite{breakingFp1} but there is few differences.
For simplicity, and because their implementations are not the same, we distinguish between operations of settlement as a buyer and as a seller: we call \textit{settle} settlements from a buyer and \textit{redeem} settlements from the seller. 
Moreover, we only consider one type of funds and \say{partially validated funds} are tracked using the counter of payments made to each seller.

\section{\protocolname: Full protocol}\label{sec-a:solution}
We detail \emph{\protocolname}, our new protocol for fractional spending resilient to adaptive adversaries.
\protocolname follows the same high-level steps as the protocol of~\cite{breakingFp1} but uses a secret quorums $\textit{sq}$ (presented in Section \ref{sec:sq}), implemented with a Ring VRF scheme $\textit{rVRF}$ (presented in Section \ref{sec:implem}), to protect the privacy of validators in the $(k_1, k_2)-$quorum system. 
It solves the fractional spending problem with $s_1=k_1$ parallel transactions and each transaction spends $\frac{1}{s_2}=\frac{1}{k_2}$ of the fund's balance, as defined in Section \ref{sec:fractionalproblem}.
Contrarily to~\cite{breakingFp1}, we batch settlements and redeems of multiple funds by using a set as parameter of these operations.

We start by Algorithm \ref{alg:valSelect} to complete the implementation of secret quorum. 
It determines $V_k$ unique validators based on a data $d$, used to select which validators will test with the VRF if they are selected or not for a quorum.
This algorithm is similar to quorum selection in \cite{breakingFp1}.
\begin{algorithm}
    \caption{Candidates selection for the VRF}
    \label{alg:valSelect}
    
    \begin{algorithmic}[1]
    \Procedure{\textit{SelectCandidates}}{$\textit{q}, V_k$}  
        \State $\textit{seed} := \textit{Hash}(\textit{q})$
        \State $\textit{candidates} := \emptyset$
        \State $j := 1$
        \While{$|\textit{candidates}| < V_k$}
            \State $v := \textit{validators}[\textit{Hash}(\textit{seed} || j)]$
            \If{$v \not \in \textit{candidates}$}
                \State $\textit{candidates} := \textit{candidates} \cup \{ v \}$
            \EndIf
            \State $j := j + 1$
        \EndWhile
    \State \Return $\textit{candidates}$ \Comment{$V_k$ unique validators based on $q$}
    \EndProcedure
    \end{algorithmic}
\end{algorithm}

\paratitle{States}
To be able to execute payments, clients are storing a set $\textit{assets}$ of funds they own. 
For later redeem, clients are storing a list of their transactions denoted $\textit{transactions}$ and $\textit{proofs}$ a map from each transaction to its secret quorum validation proof justifying its validity. 

For their validations duty, validators are storing a set $\textit{funds}$ of all funds in the system, initially set to some common arbitrary value, a set $\textit{seen}$ storing fund identifiers which they already validated a transaction from, a set $\textit{settled}$ storing fund identifiers that were already settled, a set $\textit{redeemed}$ storing transaction identifiers that were already redeemed and a map $\textit{signatures}$ from fund identifiers to the ring signatures they produced. 
Every node also store a list of all validators public keys $\textit{validators}$ used to generate ring signatures and select candidates.

\paratitle{Payments}
We present the protocol to achieve a payment between clients buyer and seller in Algorithms \ref{alg:paybuyer} and \ref{alg:paysellercomplete}. 
We denote $\langle m \rangle_{\textit{pk}}$ a message $m$ that is signed using the private key matching the public key $\textit{pk}$. 
We suppose that signatures are always verified and omit verification steps from algorithms.
The buyer signs a transaction $\textit{tx}$ with a fund and the public key of the seller, and sends $\textit{tx}$ to the seller.

\begin{algorithm}
    \caption{Payment protocol executed by Buyer}
    \label{alg:paybuyer}

    \begin{algorithmic}[1]
        \Procedure{$\textit{Pay}_{\textit{buyer}}$}{$\textit{fund}, \textit{pks}$}
        \State $\textit{tx} := \langle \textit{fund}, \textit{pks} \rangle_{\textit{pkb}}$
        \State send $(\textit{PAY}, \textit{tx})$ to $\textit{pks}$
        \State wait until $(\textit{CONFIRMPAY}, \textit{tx})$ received from $\textit{pks}$
        \EndProcedure
    \end{algorithmic}
\end{algorithm}

The seller uses $\textit{sq}.\textit{query}$ to check the validity of $\textit{tx}$ with a secret $(k_1, k_2)-$quorum of size $m$. 
Secretely selected validators check the validity of the transaction using
\[
\textit{valid}(\textit{tx}) := \textit{tx}.\textit{fund} \not \in (\textit{seen}\cup \textit{settled}) \land \textit{pkb}\in\textit{funds}[\textit{tx}.\textit{fund}].\textit{owner}
\]
If the transaction is valid, they mutate their state to take it into account with \textit{mutate}(\textit{tx}) : 
\[ 
\textit{seen} := \textit{seen} \cup \{ \textit{tx}.\textit{fund} \} \land
\textit{signatures}[\textit{tx}.\textit{fund}] := (\textit{tx}, r, N)
\]

A corrupted seller may never send a confirmation to the buyer. 
After waiting for a timeout, the buyer can assume that the payment was successful and if not he will learn about the amount not spent when executing a settlement. 
Moreover, the buyer has no way to check if a seller claiming that a payment failed is saying the truth. This can be fixed by making validators send ring signatures to both the seller and the buyer. 
For simplicity and because this issue is not considered in \cite{breakingFp1}, we omit this detail in the pseudo code. 

\begin{algorithm}
    \caption{Payment protocol executed by Seller}
    \label{alg:paysellercomplete}
    
    \begin{algorithmic}[1]
    \Procedure{$\textit{Pay}_{\textit{seller}}$}{(PAY, $\langle \textit{tx} \rangle_{\textit{pkb}}$)} \Comment{upon reception of a PAY message}
        \State $p := \textit{sq}.\textit{query}(\langle tx \rangle_{pkb}, m)$
        \If{$\textit{sq}.\textit{verify}(p, \langle tx \rangle_{pkb})$} \Comment{\textit{tx} is valid for a (secret) $(k_1, k_2)-$quorum}
            \State $\textit{proofs}[tx] := p$
            \State $\textit{transactions} := \textit{transactions} \cup \{ \textit{tx} \}$
            \State send $(\textit{CONFIRMPAY}, \textit{tx})$ to $\textit{pkb}$
        \Else
            \State \Return ERROR
        \EndIf
        \EndProcedure
    \end{algorithmic}
\end{algorithm}

\paratitle{Settlements}
Algorithm \ref{alg:fundsSettleBuyer} and \ref{alg:fundsSettleValidators} describe the protocol to achieve settlements. 
After receiving a request to settle a set of funds\footnote{Unlike \cite{breakingFp1}, we merge settlements of multiple funds to amortize the cost.} from a buyer, each validator sends to others validators the ring signatures generated when validating a transaction from a fund in $S_F$, if they did.
After receiving $n-f$ such data from others, each validator locally reconstructs the transactions and the corresponding signatures attesting their validity. 
For each valid transaction, he will subtract the transaction from the sum of the balance of each fund in $S_F$. 
The unspent balance is credited on a new fund $F'$ and its data is sent to the buyer in a confirmation message. 
After receiving $n-f$ same $F'$, the buyer can add $F'$ to its \textit{assets} and remove settled ones. 
\begin{algorithm}
    \caption{Settlement for the buyer}
    \label{alg:fundsSettleBuyer}

    \begin{algorithmic}[1]
        \Procedure{$\textit{Settle}_{\textit{buyer}}$}{$S_{F}$}
        \State send $(\textit{SETTLE}, S_{F})$ to all validators 
        \State wait until $n-f$ same $(\textit{CONFIRMSETTLE}, F')$ received
        \State $\textit{assets} = \textit{assets} \setminus S_{F} \cup \{ F' \}$
        \EndProcedure
    \end{algorithmic}

\end{algorithm}
\begin{algorithm}
    \caption{Settlement for validators}
    \label{alg:fundsSettleValidators}

    \begin{algorithmic}[1]
        \Procedure{$\textit{Settle}_{\textit{validator}}$}{$(\textit{SETTLE}, S_{F})$} \Comment{upon reception of a SETTLE message from $\textit{pkb}$}
        \State $\textit{records} := \emptyset$
        \For{$F \in S_F$}
            \If{$F\in \textit{seen} \land F\not\in \textit{settled}$}
                \State $\textit{records} := \textit{records} \cup \{ \textit{signatures[F]}\}$
            \EndIf
        \EndFor
        
        \State gossip $(\textit{SETTLERECORD}, S_{F}, \textit{records})$ to all validators
        \State wait for $n-f$ \textit{SETTLERECORD} from validators
        
        \State $\textit{records} := \textit{records} \cup \{\textit{records'}: (\textit{SETTLERECORD}, S_{F}, \textit{records'})$ received$\}$

        \State $\textit{txs} := \{\textit{tx}: (\textit{tx}, r, N) \in \textit{records} \}$
        \State $\textit{proofs} := \{(\textit{tx}, \{ r: (\textit{tx}, r, N)\in \textit{records} \}: tx \in \textit{txs}\}\}$ 

        \State $\textit{valids} := \emptyset$
        \For{$(\textit{tx}, \textit{rings}) \in \textit{proofs}$}
            \State $N := $ majority in $\textit{proofs}[\textit{tx}]$
            \If{$\textit{sq.}\textit{verify}(\{\textit{rings}, N\},\textit{tx})$} \Comment{\textit{tx} is valid for a (secret) $(k_1, k_2)-$quorum}
                \State $\textit{valids} := \textit{valids} \cup \{ \textit{tx} \}$
            \EndIf
        \EndFor

        \State $\textit{balance} := \sum_{F\in S_F} \textit{funds}[F].balance$
        \For{$\textit{tx}\in \textit{valids}$}
            \State $\textit{balance} := \textit{balance} - \textit{funds}[\textit{tx}.\textit{fund}].\textit{balance} / s_2$
        \EndFor
        
        \State $\textit{settled} := \textit{settled} \cup S_F$
        
        \If{$\textit{balance} > 0$}
            \State $funds = funds \cup \{ F' \}$ for $F'$ such that $F'.\textit{owner} = \textit{pkb}, F'.\textit{balance} = \textit{balance}$
            \State send $(\textit{CONFIRMSETTLE}, S_{F}, \textit{F'})$ to $\textit{pkb}$
        \Else
            \State send $(\textit{CONFIRMSETTLE}, S_{F}, \bot)$ to $\textit{pkb}$
        \EndIf
        \EndProcedure
    \end{algorithmic}
\end{algorithm}

\paratitle{Redeem} The protocol to redeem is presented in Algorithms \ref{alg:couponsSettleSeller} and \ref{alg:couponsSettleValidators}. 
It is a simple broadcast from the seller of the validation proof from a secret quorum collected during the validation of each transactions from a fund to redeem. 
Validators check the proofs and merge the value of each valid transaction into a single new fund. 
After receiving $n-f$ same confirmations for the same fund $F'$, $F'$ is a valid new fund that can be added to the list of \textit{assets} owned by the Seller. 
The transactions that were redeemed can be erased from the memory of the seller.

\begin{algorithm}
    \caption{Redeem for the seller}
    \label{alg:couponsSettleSeller}

    \begin{algorithmic}[1]
        \Procedure{$\textit{Redeem}_{\textit{seller}}$}{$S_F$}
        \State $\textit{txs} := \{ (\textit{tx},\textit{proofs}[\textit{tx}]) | \textit{tx} \in \textit{transactions}: \textit{tx}.\textit{fund} \in S_F \}$
        \State broadcast $(\textit{REDEEM}, txs)$ to all validators 
        \If{$n-f$ same $(\textit{CONFIRMREDEEM}, F')$ received}
            \State $\textit{assets} := \textit{assets} \cup \{ F' \}$
            \State $\textit{transactions} := \textit{transactions} \setminus \textit{txs}$
            \State\Return SUCCESS
        \ElsIf{$f+1\ (\textit{INVALIDREDEEM})$ received}
            \State\Return ERROR
        \EndIf
        \EndProcedure
    \end{algorithmic}
    
\end{algorithm}

\begin{algorithm}
    \caption{Redeem for validators}
    \label{alg:couponsSettleValidators}
    
    \begin{algorithmic}[1]
        \Procedure{$\textit{Redeem}_{\textit{validator}}$}{$(\textit{REDEEM}, \textit{txs})$} \Comment{upon reception of a REDEEM message from $\textit{pks}$}
        \State $\textit{balance} := 0$ 
        \State $\textit{valids} := \emptyset$
        \For{$(\textit{tx},\textit{proof}) \in \textit{txs}$}
            \If{$\textit{sq}.\textit{verify}(\textit{proof},\textit{tx}) \land \textit{tx}.\textit{pks} = \textit{pks} \land \textit{tx} \not\in \textit{redeemed}$}
                \State $\textit{balance} := \textit{balance} + \textit{funds}[\textit{tx}.\textit{fund}].\textit{balance} / s_2$
                \State $\textit{valids} := \textit{valids} \cup \{ \textit{tx} \}$
            \Else
                \State send $(\textit{INVALIDREDEEM})$ to $\textit{pks}$
                \State\Return ABORT
            \EndIf
        \EndFor
        \State $\textit{redeemed} := \textit{redeemed} \cup \{ \textit{tx}: (\textit{tx},\textit{proof} \in \textit{txs}) \}$

        \State $\textit{funds} := \textit{funds} \cup \{ F' \}$ for $F'$ such that $F'.\textit{owner} = \textit{pks}, F'.\textit{balance} = \textit{balance}$
        \State send $(\textit{CONFIRMREDEEM}, F')$ to $\textit{pks}$
        \EndProcedure
    \end{algorithmic}
\end{algorithm}

\section{Multi-owner without Consensus}\label{sec-a:noconsensus}
Guerraoui et al., prove in \cite{consensusNumberCryptocurrency} that a payment system --- formalized as the asset transfer problem --- where accounts are shared by $k$ owners has consensus number $k$.
We argue here that \protocolname handles unbounded shared accounts without Consensus. 
The proof from \cite{consensusNumberCryptocurrency} relies on different owners of a same account trying to send payments with \emph{different amounts} such that only one of them is valid.
This allows each owner to compute who won this hidden leader election by reading the remaining balance of the shared account.

\section{Fault-tolerance of \texorpdfstring{$(k_1, k_2)$}{(k1,k2)}-quorum Systems}\label{sec-a:ftk1k2}
Bazzi and Tucci propose in~\cite{breakingFp1} a probabilistic construction of a $(k_1, k_2)-$quorum systems, but assuming $n > 8f$.
Indeed, quorums are selected at random over $\Pi$ which implies: $n$ is large enough to ensure each selected quorums has sufficient intersection with other quorums to stop validation when the bound is reached; but also sufficient non-intersection to guarantee parallel execution when the bound is not reached. 
This construction may appear less resilient than other constructions with traditional quorum systems, which have $n > 3f$. 
This resilience, however, is global for $k$ parallel transactions.
Achieving $k$ parallel transactions using traditional quorum systems would require to split the nodes in $k$ partitions of size $m = n/k$. 
The resilience of such solution would then be reduced to the one of a unique partition: corrupting $m/3$ nodes from one quorum would allow double spending.
Thus, the probabilistic quorum systems construction from~\cite{breakingFp1} provides better resilience against corruptions than traditional quorums for $k$ greater or equal to $3$.

\end{document}

%% file: macros.tex
\newboolean{comments}
\setboolean{comments}{false}

\newcommand{\maxence}[1]{{
    \ifthenelse{
        \boolean{comments}
    }{
        \color{red}{{M: #1}}
    }
    {}
}}

\newcommand{\sara}[1]{{
    \ifthenelse{
        \boolean{comments}
    }{
        \color{blue}{{S: #1}}
    }
    {}
}}

\newcommand{\rida}[1]{{
    \ifthenelse{
        \boolean{comments}
    }{
        \color{green}{{R: #1}}
    }
    {}
}}

\newcommand{\funds}[0]{%
\ifmmode
\mathcal{F}
\else
$\mathcal{F}$
\fi
}

\newcommand{\clients}[0]{%
\ifmmode
\mathcal{P}
\else
$\mathcal{P}$
\fi
}

\newcommand{\Nat}[0]{{\mathbb N}}

\newcommand{\paratitle}[1]{{\noindent\textbf{#1.}\hspace{0.25cm}}~}

\newcommand{\protocolname}[0]{StealthDust\xspace}

\newtcolorbox{algobox}[1][]{
    enhanced,
    after skip=2mm,
    title=#1,
    breakable = true,
    fonttitle=\sffamily\bfseries,
    coltitle=cyan,
    colbacktitle=gray!10,   
    titlerule= 0pt,         
    overlay={%
        \ifcase\tcbsegmentstate
        \or%
        \else%
        \fi%
    }
    colback = gray,         
    colframe = black!75,     
    fontupper=\small, 
    fontlower=\small
    }

%% file: classicalQuorumsvsSecretQuorums.tex

\begin{tikzpicture}[
    font=\small,
    pubmsg/.style={
        -{Stealth[length=5pt]}, thick, black
    },
    adversarymsg/.style={
        -{Stealth[length=5pt]}, dashed, black
    },
    anonproof/.style={
        -{Stealth[length=5pt]}, thick, black
    },
    adversarytarget/.style={
        -{Stealth[length=4pt]}, red!70, dashed, line width=1.0pt
    },
    corruptarrow/.style={
        -{Stealth[length=5pt]}, red!80, line width=1.4pt
    },
    iconnode/.style={
        draw=red!60, fill=white, thick, circle, inner sep=2pt,
        font=\normalsize, text=red!70
    },
]

\begin{scope}[xshift=0cm]

  \draw[rounded corners=10pt, thick, draw=black!35, fill=white]
      (-4.6, -2) rectangle (4.6, 4);

  \node[font=\large\bfseries] at (0, 3.75) {Standard Small Quorum};

  \node[font=\large, text=black!60] at (-3.60, 3.05) {\faServer};
  \node[font=\large, text=black!60] at (-2.55, 2.85) {\faServer};
  \node[font=\large, text=black!60] at (-1.40, 3.10) {\faServer};
  \node[font=\large, text=black!60] at (-0.15, 2.90) {\faServer};
  \node[font=\large, text=black!60] at ( 0.95, 3.08) {\faServer};
  \node[font=\large, text=black!60] at ( 2.10, 2.88) {\faServer};
  \node[font=\large, text=black!60] at ( 3.20, 3.05) {\faServer};

  \node[font=\large, text=black!60] at (-3.55, 1.70) {\faServer};
  \node[font=\large, text=black!60] at (-2.50, 1.52) {\faServer};
  \node[font=\large, text=black!60] at ( 2.90, 1.65) {\faServer};

  \node[font=\large, text=green!60!black] (qL1) at (-1.20, 1.58) {\faServer};
  \node[font=\large, text=green!60!black] (qL2) at ( 0.10, 1.72) {\faServer};
  \node[font=\large, text=green!60!black] (qL3) at ( 1.45, 1.55) {\faServer};

  \node[font=\scriptsize, text=blue!60, align=center] at (0.2, 2.10)
      {Secret election revealed};

  \node[font=\small, text=black!80]      (envL1) at (-1.5, 1.25) {\faEnvelope};;

  \node[font=\small, text=black!80]      (envL2) at ( -0.2, 1.4) {\faEnvelope};

  \node[font=\small, text=black!80]      (envL3) at ( 1.2, 1.2) {\faEnvelope};

  \node[draw=black, thick, fill=white, rounded corners=5pt,
        minimum width=1.6cm, minimum height=1.1cm] (clientL) at (-1.5, -1.3) {};
  \node[font=\large] at (-1.5, -1.12) {\faLaptop};
  \node[font=\scriptsize\bfseries] at (-1.5, -1.55) {Client};

  \node[draw=red!60, thick, fill=red!10, rounded corners=5pt,
        minimum width=2.8cm, minimum height=1.1cm] (advR) at (2.2, -1.3) {};
  \node[font=\large, text=red!70] at (2.2, -1.10) {\faUserSecret};
  \node[font=\scriptsize, align=center] at (2.2, -1.55)
      {{\bfseries Adversary}\\[-1pt]Adaptive \& Rushing};

  \draw[pubmsg] (qL1.south) to[out=250, in=100] (clientL.north west);
  \draw[pubmsg] (qL2.south) to[out=260, in=80]  (clientL.north);
  \draw[pubmsg] (qL3.south) to[out=270, in=60]  (clientL.north east);

  \node[iconnode] (eyeL) at (1.7, -0.3) {\faEye};
 \node[font=\scriptsize, text=green!55!black] at (2.05, 0.15) {\faLockOpen};

  \draw[adversarymsg] (qL3.south) to[out=260, in=100]   (eyeL.north);
  \draw[adversarymsg] (qL2.south) to[out=250, in=100] (eyeL.north);
  \draw[adversarymsg] (qL1.south) to[out=240, in=100] (eyeL.north);

  \node[iconnode, text=red!80] (bioL) at (2.4, -0.3) {\faBiohazard};

  \draw[corruptarrow] (bioL.north) to[out=110, in=300] (qL1.east);
  \draw[corruptarrow] (bioL.north) to[out=110,  in=280] (qL2.east);
  \draw[corruptarrow] (bioL.north) to[out=80,  in=330] (qL3.east);

\end{scope}

\begin{scope}[xshift=10.0cm]

  \draw[rounded corners=10pt, thick, draw=black!35, fill=white]
      (-4.6, -2) rectangle (4.6, 4);

  \node[font=\large\bfseries] at (0,  3.75) {Secret Quorum};

  \draw[rounded corners=10pt, thick, draw=blue!55, fill=blue!4]
      (-4.3, 0.9) rectangle (4.3, 3.5);
  \node[font=\small, text=blue!70] at (0, 3.25) {Candidate set};

  \node[font=\large, text=blue!45] at (-3.60, 2.88) {\faServer};
  \node[font=\large, text=blue!45] at (-2.55, 2.70) {\faServer};
  \node[font=\large, text=blue!45] at (-1.40, 2.92) {\faServer};
  \node[font=\large, text=blue!45] at (-0.15, 2.75) {\faServer};
  \node[font=\large, text=blue!45] at ( 0.95, 2.90) {\faServer};
  \node[font=\large, text=blue!45] at ( 2.10, 2.72) {\faServer};
  \node[font=\large, text=blue!45] at ( 3.20, 2.88) {\faServer};

  \node[font=\large, text=blue!45] at (-3.55, 1.55) {\faServer};
  \node[font=\large, text=blue!45] at (-2.50, 1.38) {\faServer};
  \node[font=\large, text=blue!45] at ( 2.90, 1.48) {\faServer};

  \node[font=\large, text=green!60!black] (qR1) at (-1.20, 1.44) {\faServer};
  \node[font=\large, text=green!60!black] (qR2) at ( 0.10, 1.58) {\faServer};
  \node[font=\large, text=green!60!black] (qR3) at ( 1.45, 1.42) {\faServer};

  \node[font=\scriptsize, text=green!55!black, align=center]
      at (0.2, 2.05) {Effective quorum};

  \node[font=\small, text=black!80]      (envR1) at (-1.55, 1.15) {\faEnvelope};;

  \node[font=\small, text=black!80]      (envR2) at (-0.25, 1.3) {\faEnvelope};

  \node[font=\small, text=black!80]      (envR3) at (1.1, 1.15) {\faEnvelope};

  \node[draw=black, thick, fill=white, rounded corners=5pt,
        minimum width=1.6cm, minimum height=1.1cm] (clientR) at (-1.5, -1.3) {};
  \node[font=\large] at (-1.5, -1.12) {\faLaptop};
  \node[font=\scriptsize\bfseries] at (-1.5, -1.55) {Client};

  \node[draw=red!60, thick, fill=red!10, rounded corners=5pt,
        minimum width=2.8cm, minimum height=1.1cm] (advR) at (2.2, -1.3) {};
  \node[font=\large, text=red!70] at (2.2, -1.10) {\faUserSecret};
  \node[font=\scriptsize, align=center] at (2.2, -1.55)
      {{\bfseries Adversary}\\[-1pt]Adaptive \& Rushing};

    \node[iconnode] (eyeR) at (1.7, -0.3) {\faEye};
    \node[iconnode, text=red!80] (bioR) at (2.4, -0.3) {\faBiohazard};

  \draw[anonproof] (qR1.south) to[out=250, in=100] (clientR.north west);
  \draw[anonproof] (qR2.south) to[out=260, in=80]  (clientR.north);
  \draw[anonproof] (qR3.south) to[out=270, in=60]  (clientR.north east);

  \draw[adversarymsg] (qR1.south) to[out=260, in=100]   (eyeR.north);
  \draw[adversarymsg] (qR2.south) to[out=250, in=100] (eyeR.north);
  \draw[adversarymsg] (qR3.south) to[out=240, in=100] (eyeR.north);

  \node[font=\large\bfseries, text=red!60, align=center]
        (qmarkR) at (2.7, 0.15) {?};
  \node[font=\scriptsize, text=green!55!black] at (2.05, 0.15) {\faLock};

  \draw[adversarytarget]
      (bioR.north) to[out=120, in=270] (2.10, 2.5);
  \draw[adversarytarget]
      (bioR.north) to[out=120, in=300] (qR3.east);
  \draw[adversarytarget]
      (bioR.north) to[out=70, in=270] (2.9, 1.25);
  \draw[adversarytarget]
      (bioR.north) to[out=130, in=250] (3.20, 2.68);

  \node[font=\small\bfseries, text=red!70] at (2.30, 2.3)
      {\ding{55}};

    \node[font=\small\bfseries, text=red!70] at (3.30, 2.48)
      {\ding{55}};

    \node[font=\small\bfseries, text=red!70] at (3.1, 1.15)
      {\ding{55}};

\end{scope}

\end{tikzpicture}